\def\year{2017}\relax
\documentclass[letterpaper]{article}
\usepackage{aaai17}
\usepackage{times}
\usepackage{helvet}
\usepackage{courier}
\usepackage{enumitem}

\usepackage{amsmath,amssymb,amsthm}
\usepackage{algorithm,algorithmic}

\usepackage{tikz}
\usetikzlibrary{positioning,arrows,calc}
\usepackage{pgfplots}

\usepackage{graphicx}

\usepackage[utf8]{inputenc} 
\usepackage{url}            
\usepackage{booktabs}       
\usepackage{amsfonts}       
\usepackage{nicefrac}       
\usepackage{microtype}      

\newcommand{\argmax}{\operatornamewithlimits{argmax}}

\newcommand{\OMIT}[1]{}

\theoremstyle{definition}
\newtheorem{theorem}{Theorem}
\newtheorem{lemma}[theorem]{Lemma}

\newtheorem{corollary}[theorem]{Corollary}

\newtheorem{remark}[theorem]{Remark}
\newtheorem{example}[theorem]{Example}

\usepackage{CJKutf8}
\usepackage{ifthen}
\newcommand{\COMM}[2]{{
\begin{CJK}{UTF8}{ipxm}
\ifthenelse{\equal{#1}{TM}}{\color{blue}}{
\ifthenelse{\equal{#1}{YK}}{\color{red}}{
\ifthenelse{\equal{#1}{HS}}{\color{cyan}}{
\ifthenelse{\equal{#1}{KK}}{\color{magenta}}}}}
[#1: #2]
\end{CJK}
}}

\begin{document}
\title{Optimal Pricing for Submodular Valuations with Bounded Curvature}
\author{
Takanori Maehara\\
Shizuoka University \\
maehara.takanori@shizuoka.ac.jp
\And
Yasushi Kawase \\
Tokyo Institute of Technology \\
kawase.y.ab@m.titech.ac.jp
\AND
Hanna Sumita \\
National Institute of Informatics \\
sumita@nii.ac.jp
\And
Katsuya Tono \\
University of Tokyo \\
katsuya\_tono@mist.u-tokyo.ac.jp
\And
Ken-ichi Kawarabayashi \\
National Institute of Informatics \\
k\_keniti@nii.ac.jp
}
\nocopyright
\maketitle

\begin{abstract}
The optimal pricing problem is a fundamental problem that arises in combinatorial auctions.
Suppose that there is one seller who has indivisible items and multiple buyers who want to purchase a combination of the items.
The seller wants to sell his items for the highest possible prices,
and each buyer wants to maximize his utility (i.e., valuation minus payment) as long as his payment does not exceed his budget. 
The optimal pricing problem seeks a price of each item and an assignment of items to buyers such that every buyer achieves the maximum utility under the prices. 
The goal of the problem is to maximize the total payment from buyers. 
In this paper, we consider the case that the valuations are submodular. 
We show that the problem is computationally hard even if there exists only one buyer. 
Then we propose approximation algorithms for the unlimited budget case. 
We also extend the algorithm for the limited budget case when there exists one buyer and multiple buyers collaborate with each other. 
\end{abstract}

\section{Introduction}
\label{sec:introduction}

\subsection{Background and motivation}
In a \emph{combinatorial auction} \cite{nisan2007algorithmic11,cramton2006combinatorial},
a seller has a set of indivisible items, and buyers purchase a combination of the items.
The seller wants to sell his items to the buyers for the highest possible prices, and
each buyer wants to purchase a set of items that is valuable for him and also has a reasonable price.
More precisely, each buyer purchases a set of items that maximizes his \emph{utility} within the limits of his budget; here, the utility for a set of items is the valuation for him (minus the payment for purchase).
Thus, the seller seeks a price of each item (not bundle) and an assignment of items to buyers such that they are \emph{stable}, i.e., no buyer can gain more utility by changing the set of items that he purchases.
The goal is to maximize the total profit obtained from the buyers. 
The stability captures a fairness condition for the individual buyers \cite{goldberg2003envy,guruswami2005profit,cheung2008approximation,anshelevich2015envy}. 
In general, such a problem is called the \emph{optimal pricing problem} and studied in many situations.

In this paper, we assume that the valuations of buyers are represented by \emph{submodular} functions, which capture the notion of the diminishing returns property. 
Submodular functions appear in many situations \cite{bach2010structured,SomaY15}, and have been studied extensively. 
For simplicity, 
we assume that every buyer truthfully tells the seller his valuation and his budget.
The purpose of this paper is to analyze theoretical properties of the optimal pricing problems with submodular valuations, and to propose (approximation) algorithms for the problems. 
We analyze the performance of the algorithms in terms of the \emph{curvatures} of the valuations, which capture the degree of nonlinearity.
The curvature has been used to derive better approximation ratios for several submodular optimization problems~\cite{iyer2013submodular,iyer2013curvature,sviridenko2015optimal,vondrak2010submodularity}. 


\subsection{Our contributions}

We summarize our results for the optimal pricing problem.

We first show that the optimal pricing problem with submodular valuations is NP-hard even for instances derived from our application (Theorem~\ref{thm:NPhard}). 
Moreover, there exists an instance that requires exponentially many oracle evaluations in the oracle model (Theorem~\ref{thm:exporacle}). 

Our main result is to propose approximate pricing algorithms for the following three cases: single buyer case (Algorithm~\ref{alg:singlepricing}), multiple buyers case (Algorithm~\ref{alg:multiple}), and multiple collaborating buyers case (Algorithm~\ref{alg:collaboratepricing}).
Then we show the approximation ratios for these algorithms in terms of the curvatures of the valuations (Theorems~\ref{thm:singlepricing}, \ref{thm:multiplepricing}, \ref{thm:collaborating}). 
Our algorithms output a nearly optimal solution if the curvatures are small.
We will point out that a practical application of our problem has submodular valuations whose curvatures are typically small by using a general upper bound on curvatures (Theorem~\ref{thm:curvature}). 
This justifies our analysis of the optimal pricing problem using curvatures. 
The application is a similar problem to the budget allocation problem, which is widely studied both theoretically and practically in computational advertising. 
We include all the proofs in the Appendix.

We also conduct computational experiments on some synthetic and realistic datasets to evaluate the proposed pricing algorithms (Section~\ref{sec:experiments}). 
Our algorithm performs better than baseline algorithms.
To the best of our knowledge, no prior work proposed a suitable algorithm for the optimal pricing problem with submodular valuations.

\subsection{Related Work}
\label{sec:relatedwork}

The optimal pricing problem is also referred to as the \emph{profit maximizing pricing problem}. 
There are many works for the case when valuations of buyers are \emph{unit-demand} or \emph{single-minded}~
\cite{aggarwal2004algorithms,goldberg2003envy,goldberg2001competitive,guruswami2005profit,cheung2008approximation,anshelevich2015envy}. 
A valuation $f$ is called unit-demand if $f(X)=\max_{x \in X} f(x)$ for any $X$ with $|X| \geq 2$, and called single-minded if for some $S^*$, it holds that $f(X)=f(S^*) > 0$ for any $X \supseteq S^*$ and $f(X)=0$ otherwise. 
Any unit-demand valuation is submodular, while single-minded valuations are not necessarily submodular.  
Guruswami et al.~\shortcite{guruswami2005profit} proved the APX-hardness of the optimal pricing problem where the valuations are all unit-demand or all single-minded.
Thus, the general optimal pricing problem is computationally intractable.
They also provided logarithmic approximation algorithms under the assumption that the valuations are all unit-demand or all single-minded. 
Since these algorithms fully rely on the assumption, they do not extend to our case. 

We remark that the optimal pricing problem is different from the problem of finding a Walrasian equilibrium.
A pair of a pricing and an assignment is called a \emph{Walrasian equilibrium} (or \emph{competitive equilibrium}) if it is stable and all positive-price items are allocated to some buyer~\cite{nisan2007algorithmic11}.
In our model, the seller can decide a subset of items that are not allocated.
This difference may improve the seller's profit. 
We give such an example in Remark \ref{remark:SellAllNotOptimal}.

The \emph{winner determination problem} is also similar to the optimal pricing problem. 
This is the problem of finding an allocation of items to buyers that maximizes the sum of buyers' valuations in combinatorial auctions. 
Rothkopf et al.~\shortcite{rothkopf1998computationally} proved the NP-hardness of the winner determination problem. 
Sandholm~\shortcite{sandholm2002algorithm} provided an inapproximability result on the problem and some approximation algorithms for special cases.
For more details of this problem, see, e.g.,~\cite{nisan2007algorithmic11,cramton2006combinatorial}.
The winner determination problem maximizes the total valuation of buyers whereas our problem maximizes the profit of the seller.
In general, these problems have different optimal solutions (see Remark \ref{remark:SellAllNotOptimal} for an example). 
So our problem setting is different from the winner determination problem.

\section{Preliminaries}
\label{sec:preliminaries}

In this section, we review submodular functions and the optimal pricing problem, and describe our motivation to study the optimal pricing problem. 


\subsection{Submodular function and curvature}

Let $V$ be a finite set.
A function $f: 2^V \to \mathbb{R}$ is \emph{submodular} if it satisfies
\begin{align}
  \label{eq:submodular}
  f(X) + f(Y) \ge f(X \cap Y) + f(X \cup Y)
\end{align}
for all $X, Y \subseteq V$~\cite{fujishige2005submodular}.
This condition is equivalent to the \emph{diminishing returns property}:
$f(X) - f(X \setminus x) \ge f(Y) - f(Y \setminus x)$
for all $X \subseteq Y \subseteq V$;
here we denote ``$X \setminus \{x\}$'' by ``$X \setminus x$'' for notational simplicity. 
We say that $f$ is \emph{monotone nondecreasing} if $f(X) \le f(Y)$ for all $X \subseteq Y$.
In this paper, we assume that $f(\emptyset) = 0$.

The diminishing returns property is a fundamental principle of economics~\cite{samuelson2004microeconomics}.
Thus, submodular functions are often used to model user utilities and preferences.
They also appear in 
combinatorial optimization~\cite{fisher1978analysis,fujishige2005submodular},
social network analysis~\cite{kempe2015maximizing}, and machine learning~\cite{bach2010structured,PanJGBJ14,SomaY15}. 

For a monotone nondecreasing function $f: 2^V \to \mathbb{R}$ and an integer $s \in \{0,\ldots, |V|\}$, 
the \emph{curvature} $\kappa(s)$ is defined by the largest nonnegative number that satisfies
\begin{align}
  (1 - \kappa(s)) f(x) \le f(X) - f(X \setminus x)
\end{align}
for all $|X| = s$ and $x \in X$~\cite{conforti1984submodular}.%
\footnote{Originally, Conforti and Cornu{\'e}jols~\shortcite{conforti1984submodular} introduced \emph{total curvature} and \emph{greedy curvature} for monotone nondecreasing submodular functions.}
If $f$ is submodular, then its curvature $\kappa(s)$ is a monotone nondecreasing sequence by the diminishing returns property.

We remark that computing $\kappa(s)$ is difficult since it requires exponentially many function evaluations;
therefore, we cannot use explicitly the value in an algorithm.

\subsection{Optimal pricing problem}

Here we define the optimal pricing problem.
Suppose that a seller wants to sell indivisible items $V$ to buyers $N = \{1, \ldots, n\}$ simultaneously.
Each buyer $i \in N$ has a budget $B_i$ and a valuation function $f_i:2^V\to\mathbb{R}$, where $f_i$ is a monotone nondecreasing submodular function. 
We denote by $\kappa_i$ the curvature of $f_i$ for $i \in N$. 
In this problem, we find a price vector $p \in \mathbb{R}^V$ and an assignment $(X_1, \ldots, X_n)$ which is a subdivision of $V$. 
For a price vector $p$ and an item set $X \subseteq V$, let $p(X) = \sum_{x \in X} p(x)$. 
Buyers are assumed to have quasi-linear utility, i.e., the utility of $i\in N$ is given by $f_i(X_i) - p(X_i)$.

The seller wants to maximize the total profit $p(X_1) + \cdots + p(X_n)$.
On the other hand, each buyer also wants to maximize his utility, as long as his payment to the seller does not exceed his budget.
Therefore, the assignment must satisfy some ``agreement'' condition.
We say that a price vector $p$ and assignment $(X_1, \ldots, X_n)$ pair is \emph{stable}
if it satisfies
\begin{align}
  \label{eq:stable}
  f_i(X_i) - p(X_i) \ge f_i(X) - p(X)
\end{align}
for any $i \in N$ and all $X \subseteq V$.
The stability condition means that each buyer $i$ has no incentive to change his allocation $X_i$ under the pricing $p$.
For a price vector $p$, we define the \emph{demand set} of buyer $i$ as a family of sets $X$ satisfying \eqref{eq:stable}, denoted by 
\begin{align}\label{eq:demandset}
  D_i(p) = \argmax \{ f_i(X) - p(X) \mid X \subseteq V, \ p(X) \leq B_i \}.
\end{align}
The stability condition is necessary to avoid a grudge or an antipathy of buyers even when each buyer $i$ knows his own allocated items $X_i$ and every price of the items $V$.

The \emph{optimal pricing problem} 
seeks a price vector $p\in \mathbb{R}^V$ and an assignment $(X_1, \ldots, X_n)$ that maximizes the total profit $p(X_1) + \cdots + p(X_n)$ under the stability condition. 
It is formulated as
\begin{align}
    \text{maximize}&\textstyle \quad \sum_{i\in N}p(X_i)  \nonumber \\
    \text{subject to}&\quad X_i \in D_i(p) ~(i \in N), \label{eq:problemmultiple} \\
    &\quad X_i \cap X_j = \emptyset ~(i \neq j). \nonumber
\end{align}
We propose algorithms for the problem \eqref{eq:problemmultiple} where all buyers have unlimited budgets, i.e., $B_i = +\infty$ for all $i \in N$. 
We extend our results to the limited budget case (see~\cite{maehara2017optimal}). 

\subsection{Application}
\label{subsec:budgetallocation}


We present an application that motivates us to study the optimal pricing problem with submodular valuations. 
We will claim that the curvatures of valuations are small in practice. 


Consider that the publisher (= seller) has a set $V$ of marketing channels, and that there is a set $N$ of advertisers (= buyers) that have budgets $B_i \ (i \in N)$ . 
Each advertiser $i \in N$ purchases a subset $X$ of channels for advertising under the budget constraint, i.e., $p(X) \leq B_i$. 
The valuation $f_i(X)$ of $X \subseteq V$ for advertiser $i$ is the expected value of the total revenue from loyal customers influenced by marketing channels in $X$. 
Let $p \in \mathbb{R}^V$ be the price vector, i.e., $p(v)$ is the price to publish an advertisement through marketing channel $v$. 
Each advertiser $i$ wants to buy $X_i$ that maximizes the total revenue minus the cost, i.e., $f_i(X_i) - p(X_i)$, under the budget constraint $p(X_i) \leq B_i$. 
In the following, we explicitly formulate the valuation function $f_i$ of each advertiser $i$. 

We adopt the bipartite influence model of advertising proposed by Alon et al.~\cite{alon2012optimizing}. 
Let $W$ be set of customers.
We consider a bipartite graph $G = (V \cup W, E)$. Each edge $(v, w) \in E$ indicates that marketing channel $v$ affects customer $w$.
Each edge $(v,w)$ is assigned probabilities $q_i(v,w) \ (i \in N)$, called \emph{activation probability}. 
If advertiser $i$ puts an advertisement on marketing channel $v$, then customer $w$ will become a loyal customer of buyer $i$ with probability $q_i(v,w)$. 

The probability $Q_i(X,w)$ that customer $w$ becomes a loyal customer 
when a advertiser $i$ runs advertisements on $X \subseteq V$ is given by
$Q_i(X, w) = 1 - \prod_{x \in X, (x,w) \in E}(1 - q_i(x,w))$.
Thus, the expected number of his loyal customers is $\sum_{w \in W} Q_i(X, w)$.
Let $\gamma_i$ be the expected revenue from one loyal customer. 
The expected total revenue is given by 
\begin{align}\label{eq:fQ}
\textstyle
f_i(X)=\gamma_i \cdot \sum_{w \in W} Q_i(X, w).
\end{align}
Since $Q_i(X, w)$ is a monotone nondecreasing submodular function in $X$,
$f_i(X)$ is also a monotone nondecreasing submodular function.

Here, we can observe that each curvature $\kappa_i$ of $f_i$ is small ($i \in N$). 
This implies that our analysis based on the curvature is particularly effective for this application. 
\begin{lemma}
\label{lem:curvaturebound}
For each $i \in N$ and $s \in \{1, \ldots, |V|\}$,
$\kappa_i(s) \le \max_{|X| = s,~ x \in X,~ (x,w) \in E} Q(X \setminus x, w)$. 
\end{lemma}
\begin{theorem}
\label{thm:curvature}
For each $i \in N$ and $s \in \{1, \ldots, |V|\}$, 
if $q_i(e) \le q$ for any $e \in E$, then
it holds that $\kappa_i(s) \le 1 - (1 - q)^{\min\{s,d\} - 1}$,
where $d$ is the maximum degree of the right vertices $W$.
\end{theorem}
In practice, $d$ is relatively small (e.g., $d \le 100$) since it is the number of incoming information channels of a customer.
Moreover, $q$ is very small (e.g., $q \le 0.001$) since it is the probability of gaining a customer through a single advertisement.

\section{Single buyer}
\label{sec:single}

In this section, we analyze the optimal pricing problem with a single buyer (i.e., $n=1$). 
We prove the NP-hardness of the problem and present a nearly optimal approximate algorithm for the buyer with an unlimited budget. 

Consider that there is only one buyer with an unlimited budget. 
For notational convenience, let $f: 2^V \to \mathbb{R}$ be his valuation, which is a monotone nondecreasing submodular function. 
For a price vector $p$, we denote by $D(p)$ the demand set for $p$.  
When we fix an assignment, we can easily obtain the maximum profit for the assignment. 
\begin{lemma}\label{lem:single_obs}
Let $X$ be an assignment. 
An optimal price vector for \eqref{eq:problemmultiple} with fixed $X$ is given by 
$p(x) = f(X)-f(X \setminus x)$ if $x \in X$, and $p(x)=+\infty$ otherwise. 
\end{lemma}
From this lemma, we obtain the following characterization of optimal solutions to \eqref{eq:problemmultiple}. 
\begin{lemma}\label{lem:Dtoh}
Let $X$ be an assignment. 
There exists a price vector $p$ such that $(p, X)$ is optimal to  \eqref{eq:problemmultiple} if and only if $X$ achieves
\begin{align}\label{eq:problemsingle_eq}
\textstyle
  \max_{X' \subseteq V} \sum_{x \in X'} (f(X') - f(X' \setminus x)). 
\end{align}
\end{lemma}
This implies that problem \eqref{eq:problemmultiple} is equivalent to \eqref{eq:problemsingle_eq}. 
For any $X \subseteq V$, we denote \(h(X) = \sum_{x \in X} (f(X) - f(X \setminus x))\).

We show the NP-hardness of \eqref{eq:problemmultiple} by reducing the \emph{one-in-three positive 3-SAT problem}. 
Given a boolean formula in conjunctive normal form with three positive literals per clause, the one-in-three positive 3-SAT problem determines whether there exists a truth assignment to the variables so that each clause has exactly one true variable. 
This problem is known to be NP-complete~\cite{schaefer1978complexity}.

\begin{theorem}
\label{thm:NPhard}
Problem \eqref{eq:problemmultiple} is NP-hard even when the function $f$ is of the form \eqref{eq:fQ}.
\end{theorem}

Furthermore, we obtain the result below.
\begin{theorem}
\label{thm:exporacle}
If $f$ is given by an oracle, problem~\eqref{eq:problemmultiple} requires exponentially many oracle evaluations.
\end{theorem}

Since \eqref{eq:problemmultiple} is NP-hard, we propose an algorithm to find an approximate pricing. 
Once we determine an assignment, an optimal price vector for the assignment is easily obtained from Lemma \ref{lem:single_obs}. 
Thus, we only need to find an assignment $X$ maximizing $h(X)$. 
However, an overly large assignment $X$ may have small $h(X)$ value. 
In our algorithm, we assign the top $s$ elements in order of their value, for each $s =1, \ldots, |V|$. 
The formal description is given in Algorithm~\ref{alg:singlepricing}.  
\begin{algorithm}[tb]
\caption{Pricing algorithm for a single buyer.}
\label{alg:singlepricing}
\begin{algorithmic}
  \STATE{\textbf{For} $s = 1, 2, \ldots, |V|$}
  \STATE{\quad Let $X^s$ be the largest $s$ elements of $f(x)$}
  \STATE{\quad Price $p^s(x) = f(X^s) - f(X^s \setminus x)$ $(x \in X^s)$ and $p^s(y) = +\infty$ $(y \in V \setminus X^s)$}
  \STATE{\textbf{Return} $p$ and $X$ that attains maximum of $p^s(X^s)$}
\end{algorithmic}
\end{algorithm}

This algorithm can be implemented to run in $O(|V| \log |V| + A |V|^2)$ time,
where $A$ is the computational cost of evaluating $f(X)$.
For a variant of budget allocation problem with the bipartite graph model, 
if we implement the algorithm carefully, it runs in $O(|V| |E|)$ time.

We analyze the approximation ratio of our algorithm.
\begin{theorem}
\label{thm:singlepricing}
Let $(p^*, X^*)$ be an optimal solution to \eqref{eq:problemmultiple}, and let $(p, X)$ be the output of Algorithm~\ref{alg:singlepricing}.
Then, it holds that $X \in D(p)$ and \((1 - \kappa(|X^*|)) p^*(X^*) \le p(X)\).
\end{theorem}

\begin{remark}\label{remark:SellAllNotOptimal}
  Selling all items (i.e., $s = |V|$) is not always optimal, even when the function $f$ is of the form \eqref{eq:fQ}. 
To demonstrate this, let us consider the instance of the application in Section \eqref{subsec:budgetallocation} where there are two channels $u, v$ and one user $w$ with $\gamma = 1$. 
The activation probabilities from $u$ and $v$ to $w$ are $0.9$.
When we use the both channels, the activation probability is $1-(1-0.9)^2=0.99$.
Thus $f(u) = f(v) = 0.9$ and $f(u,v) = 0.99$.
The optimal pricing sells only a single channel $X^* = \{u\}$ at price $p^*(u) = 0.9$, and $p^*(X^*) = 0.9$.
On the other hand, to sell all items $X' = \{u, v\}$, the price should be $p'(u) = p'(v) = 0.99-0.9=0.09$, and hence $p'(X') = 0.18$.

We also remark that this example shows the difference between our problem and related problems, namely, the problem of finding Walrasian equilibrium and the winner determination problem. 
There are two Walrasian equilibria $(p', X')$ and $(p'', X')$ where $p''(u)=p''(v)=0$. 
Thus $(p', X')$ achieves the maximum profit $p'(X')=0.18$ among Walrasian eqiulibria whereas the optimal value for our problem is $0.9$. 
When we regard this example as an instance of the winner determination problem, the optimal solution is $X'$ and its valuation of the \emph{buyer} is $0.99$. 
However, the optimal solution for our problem sells only $X^*$, and the profit of the seller is $0.9$. 
\end{remark}

We also show that if the curvature of $f$ is small, then the optimal values of \eqref{eq:problemmultiple} and the one without the stability condition is almost the same; see 
Theorem \ref{thm:muchworse} and \ref{thm:worsebutcurvature} in Appendix.

\section{Multiple buyers}
\label{sec:multiple}

In this section, we deal with the general optimal pricing problem \eqref{eq:problemmultiple} that admits more than one buyer. 
Recall that if $n=1$, then for any assignment $X$, there always exists a price vector $p$ satisfying $X \in D(p)$ (see Lemma \ref{lem:single_obs}).
However, in general, there may not exist a price vector $p$ such that $X_i \in D_i(p)$ for some assignment $(X_1,\dots,X_n)$.
Moreover, it is difficult to determine whether or not such a price vector exists for a given assignment.

We first approach the coNP-hardness of deciding the existence of a stable price vector for a given assignment by reducing the \emph{exact cover by 3-sets problem (X3C)}, which is NP-complete~\cite{garey1979computers}.
  In this problem, we are given a set $E$ with $|E|=3l$ and a collection $\mathcal{C}=\{C_1,\dots,C_m\}$ of 3-element subsets of $E$. 
  The task is to decide whether or not $\mathcal{C}$ contains an exact cover for $E$, i.e., a subcollection $\mathcal{C}'\subseteq \mathcal{C}$ such that every element of $E$ occurs in exactly one member of $\mathcal{C}'$.
\begin{theorem}
\label{thm:stablepricing}
  It is coNP-hard to determine,
  for a given assignment $(X_1,\dots,X_n)$,
  the existence of price vector $p$ such that
  $X_i\in D_i(p)$ for all $i\in N$.
\end{theorem}

We also show that, given a price vector $p$, it is NP-hard to decide the existence of an assignment $\mathbf{X} = (X_1,\dots,X_n)$ such that $(p, \mathbf{X})$ is stable 
(Theorem \ref{thm:NPhardmultiple} in Appendix). 

By above results, it is difficult to find a stable pair $(p, \mathbf{X})$ for given $p$ (or $\mathbf{X}$). 
Therefore in order to obtain efficiently an approximate solution, we take a natural approach that we slightly relax the stability condition. 

For any positive number $\alpha \leq 1$ and each buyer $i$, we define the $\alpha$-demand set of buyer $i$ as 
$D^{\alpha}_i(p)=\left\{X\subseteq V\mid  f_i(X)-p(X)\ge \alpha f_i(Y)-p(Y), \  \forall Y\subseteq V \right\}$. 
For a price vector \(p\) and an assignment \(\mathbf{X}=(X_1,\dots,X_n)\),
we say that \((p,\mathbf{X})\) is \emph{$\alpha$-stable}
if \(X_i\in D^\alpha_i\) for all $i\in N$.

\begin{algorithm}[tb]
\caption{Pricing algorithm for multiple buyers}
\label{alg:multiple}
\begin{algorithmic}
  \STATE{\textbf{For} $s = 1, 2, \ldots, |V|$}
  \STATE{\quad Let $X^s$ be $s$ largest elements of $\max_{i\in N} f_i(x)$}
  \STATE{\quad Price $p^s(x) = \max_{i\in N} (f_i(X^s) - f_i(X^s \setminus x))$ $(x \in X^s)$ and $p^s(y) = +\infty$  $(y \in V \setminus X^s)$}
  \STATE{\quad Let $(X^s_1,\dots,X^s_n)$ be a partition of $X^s$ such that $p^s(x) = f_i(X^s) - f_i(X^s \setminus x)~(\forall x \in X^s)$ }
  \STATE{\textbf{Return} $p$ and $(X_1,\dots,X_n)$ that attains maximum of $\sum_{i\in N}p^s(X^s_i)$}
\end{algorithmic}
\end{algorithm}
We propose a pricing algorithm in Algorithm~\ref{alg:multiple}.
The algorithm can be implemented to run in $O(A n |V|^2 + |V| \log |V|)$ time,
where $A$ is the computational cost of evaluating $f_i(X)$ ($i \in N$).
It has the following theoretical guarantee.
Here we denote \(\kappa(s)=\max \kappa_i(s)\).
\begin{theorem}
\label{thm:multiplepricing}
Let $p^*$ and $\mathbf{X}^*=(X_1^*, \ldots, X_n^*)$ be the optimal solution to \eqref{eq:problemmultiple}
and let $p$ and $\mathbf{X}=(X_1, \ldots, X_n)$ be the solution obtained by Algorithm~\ref{alg:multiple}.
Then, for $s = |X_1 \cup \cdots \cup X_n|$ and $s^* = |X_1^* \cup \cdots \cup X_n^*|$,
\((p,\mathbf{X})\) is $(1-\kappa(s))$-stable and 
$  (1 - \kappa(s^*)) \sum_{i\in N} p^*(X_i^*) \le \sum_{i\in N} p(X_i)$.
\end{theorem}

\section{Multiple collaborating buyers}
\label{sec:collaborate}

In this section, we analyze the optimal pricing problem with \emph{collaborating buyers}, i.e., the case where buyers cooperate to maximize the total utilities.
This occurs when buyers are employed by the same organization.
We present an approximation algorithm for this problem. 

We first describe the model. 
Assume that there are buyers $N = \{1, \ldots, n\}$, whose valuation functions are given by $f_1, \ldots, f_n$. 
Let $(X_1, \ldots, X_n)$ be an assignment. 
Since the goal of buyers is to maximize the sum of their utilities, the stability condition is written as
$\sum_{i\in N}(f_i(X_i)-p(X_i)) \ge \sum_{i\in N}(f_i(Y_i)-p(Y_i))$ for any assignment $(Y_1, \ldots, Y_n)$. 

Because only the total amount of the utilities matters, the publisher only needs to find a set $X$ and a price vector $p$ that satisfies the above stability condition for some partition of $X$. 
Thus, in the following, we assume that there exists one buyer who
represents the set of original buyers.
Let $f: 2^V \to \mathbb{R}$ be an aggregated valuation function defined by
$f(X) = \max_{(X_1,\dots,X_n):\text{ partition of } X}  \sum_{i\in N}f_i(X_i)$
for $X \subseteq V$. 
Note that $f(X)$ is monotone nondecreasing but not necessarily submodular 
(See Example~\ref{ex:nonsubmodular} in Appendix).

By using the aggregated valuation function, 
the stability condition is equivalent to the condition that $ f(X) - p(X) \ge f(Y) - p(Y)$
for all $Y \subseteq V$. 
Thus, the demand set is defined as \eqref{eq:demandset}, and the optimal pricing problem for collaborating buyers is formulated as
\begin{align}
  \label{eq:problemcollaborating}
    \text{maximize} \; \; p(X) \quad
    \text{subject to} \; \; X \in D(p).
\end{align}

Although the aggregated valuation function 
is not necessarily submodular, problem \eqref{eq:problemcollaborating} has a similar formulation to \eqref{eq:problemmultiple} with a single buyer. 
We obtain a similar result to Lemma \ref{lem:single_obs}. 
\begin{lemma}\label{lem:collaborating_obs}
Let $X$ be an assignment. 
An optimal price vector for \eqref{eq:problemcollaborating} with fixed $X$ is given by 
$p(x) =  \min_{Y\subseteq X: x\in Y}(f(Y\cup x) - f(Y))$ if $x \in X$, and $p(x) = +\infty$ otherwise.
\end{lemma}

From this lemma, we see that problem \eqref{eq:problemcollaborating} is equivalent to the problem of finding $X \subseteq V$ maximizing $\sum_{x \in X} \min_{Y\subseteq X: x\in Y}(f(Y) - f(Y \setminus x))$. 
Thus, we can apply the same principles as the ones of Algorithm \ref{alg:singlepricing}.
In fact, setting prices $p^s(x)=\min_{Y\subseteq X^s: x\in Y}(f(Y)-f(Y\setminus x))$ implies a similar result.
However, computing $f(X)$ is intractable (this problem is called submodular welfare problem)
and hence 
it is hard to evaluate the value \(\min_{Y\subseteq X: x\in Y}(f(Y)-f(Y\setminus x))\).
Thus, we need a further modification.

Our algorithm, summarized in Algorithm \ref{alg:collaboratepricing},
finds an approximate solution to \eqref{eq:problemcollaborating}
in $O(A n |V|^2 + |V| \log |V|)$ time,
where $A$ is the computational cost of evaluating $f_i(X)$ ($i \in N$).
We analyze the approximation ratio of our algorithm.
Let $\kappa_1, \ldots, \kappa_n$ be curvatures of $f_1, \ldots, f_n$
and \(\kappa(s) = \max_j \kappa_j (s)\) for $s=1,\dots,|V|$.

\begin{algorithm}[tb]
\caption{Pricing algorithm for collaborating buyers}
\label{alg:collaboratepricing}
\begin{algorithmic}
  \STATE{\textbf{For} $s = 1, 2, \ldots, |V|$}
  \STATE{\quad Let $X^s$ be $s$ largest elements of $f(x)~(=\max_{i} f_i(x))$}
  \STATE{\quad Price $p^s(x) = \min_{i\in N} \frac{f_i(X^s)-f_i(X^s\setminus x)}{f_i(x)} f(x)$ $(x \in X^s)$ and $p(y) = +\infty$ $(y \in V \setminus X^s)$}
  \STATE{\textbf{Return} $p$ and $X$ that attains maximum of $p^s(X^s)$}
\end{algorithmic}
\end{algorithm}

\begin{theorem}\label{thm:collaborating}
Let $(p^*, X^*)$ be an optimal solution to \eqref{eq:problemcollaborating}, 
and let $(p, X)$ be the output of Algorithm~\ref{alg:collaboratepricing}.
It then holds that \(X\in D(p)\) and 
\((1 - \kappa(|X^*|)) p^*(X^*) \le p(X)\).
\end{theorem}

To prove this theorem, we show the following two lemmas. 
\begin{lemma}\label{lem:collaborating 1}
For a set $X$ and $x\in X$, it holds that
$ f(X) - f(X \setminus x) \le f(x)$. 
\end{lemma}

\begin{lemma}\label{lem:collaborating 2}
For a set $X$ and $x\in X$, it holds that
$f(X) - f(X \setminus x) \ge \min_{i\in N}\frac{f_i(X)-f_i(X\setminus x)}{f_i(x)}f(x)
\ge (1 - \kappa(|X|)) f(x)$. 
\end{lemma}

\section{Experiments}
\label{sec:experiments}

In this section, we present experimental results on our pricing algorithms
for a variant of budget allocation problem, which are described in Section~\ref{subsec:budgetallocation}.
All experiments were conducted on an Intel Xeon E5-2690
2.90GHz CPU (32 cores) with 256GB memory running Ubuntu 12.04.
All codes were implemented in Python 2.7.3.

We performed the following five experiments:
For the single advertiser case,
(1) we computed prices of each channel for a realistic dataset;
(2) we compared the proposed algorithm with other baseline algorithms;
(3) we evaluated the scalability of the proposed algorithm; and
(4) we observed the relationship between the activation probabilities and the number of allocated channels. 
For the multiple advertisers case and the multiple collaborating advertisers case,
(5) we observed the relationship between the obtained profit and the number of advertisers.


%
For these experiments, we used two random synthetic networks (\texttt{Uniform}, \texttt{PowerLaw}) and three networks constructed from real-world datasets (\texttt{Last.fm}, \texttt{MovieLens}, \texttt{BookCrossing}).
Throughout the experiments, we assume that the expected revenue from one loyal customer is $1$, i.e., $\gamma_i=1$.
The description of the datasets is given in Appendix.

\vspace{-1em}
\paragraph{(1) Typical result}

First, we ran Algorithm~\ref{alg:singlepricing}
to \texttt{Last.fm} dataset to compute prices for the musics played in Last.fm.
Top 10 frequently played musics and top 10 high price musics
are displayed in Table~\ref{tbl:rank_by_frequency} and Table~\ref{tbl:rank_by_price}, respectively.
We can observe that some musics with a large number of plays (or unique users) are not assigned high prices.
This occurs because of the stability condition.

\begin{table}[tb]
  \centering
  \caption{Ranking by \#plays.}
  \label{tbl:rank_by_frequency}
  \resizebox{\columnwidth}{!}{%
  \setlength{\tabcolsep}{2pt}
  \begin{tabular}{clll}
    rank & artist -- music & \#play & UU \\ \hline
    1 & The Postal Service -- Such Great Heights & 3992  & 321\\
    2 & Boy Division -- Love Will Tear Us Apart  & 3663  & 318\\
    3 & Radiohead -- Karma Police                & 3534  & 346\\
    4 & Muse -- Supermassive Black Hole          & 3483  & 263\\
    5 & Death Cab For Cutie -- Soul Meets Body   & 3479  & 233\\
    6 & The Knife -- Heartbeats                  & 3156  & 177\\
    7 & Muse -- Starlight                        & 3060  & 260\\
    8 & Arcade Fire -- Rebellion (Lies)          & 3048  & 292\\
    9 & Britney Spears -- Gimme More             & 3004  & 59\\
   10 & The Killers -- When You Were Young       & 2998  & 235\\
    \hline
  \end{tabular}
  }
\end{table}
\begin{table}[tb]
  \centering
  \caption{Ranking by prices.}
  \label{tbl:rank_by_price}
  \resizebox{\columnwidth}{!}{%
  \setlength{\tabcolsep}{2pt}
  \begin{tabular}{ccll}
    rank & original & artist -- music & price \\ \hline
    1 &  1 & The Postal Service -- Such Great Heights  & 3.330  \\
    2 &  8 & Arcade Fire -- Rebellion (Lies)           & 2.101  \\
    3 &  4 & Muse -- Supermassive Black Hole           & 2.029  \\
    4 & 11 & Interpol -- Evil                          & 2.026  \\
    5 &  3 & Radiohead -- Karma Police                 & 2.003  \\
    6 &  6 & The Knife -- Heartbeats                   & 1.992  \\
    7 & 12 & Kanye West -- Love Lockdown               & 1.893  \\
    8 & 17 & Arcade Fire -- Neighborhood \#1 (Tunnels) & 1.868  \\
    9 & 23 & Kanye West -- Heartless                   & 1.788  \\
   10 & 24 & Radiohead -- Nude                         & 1.770  \\
    \hline
  \end{tabular}
  }
\end{table}

\vspace{-1em}
\paragraph{(2) Comparison with other pricing algorithms}

\begin{table}[tb]
\caption{Comparison of pricing algorithms on several datasets. Each value is the ratio of the profit obtained by the algorithm and the proposed algorithm.}
\label{tbl:comparison}
\centering
\scalebox{0.75}{
\begin{tabular}{l|ccccc}
                   & Proposed & SellAll & Random & Scaled & Ascend  \\ \hline
\texttt{Uniform} & \textbf{1.00}     & 0.89    & 0.55   & 0.98   & 0.96    \\
\texttt{PowerLaw} & \textbf{1.00}     & 0.89    & 0.65   & 0.98   & 0.51    \\
\texttt{Last.fm} & \textbf{1.00}     & 0.71    & 0.46   & 0.99   & 0.78    \\
\texttt{MovieLens} & \textbf{1.00}     & 0.58    & 0.48   & 0.96  & 0.67    \\
\texttt{BookCrossing} & \textbf{1.00}     & \textbf{1.00}    & 0.39   & 0.78   & 0.43    \\
\hline
\end{tabular}
}
\end{table}

Next, we compared Algorithm~\ref{alg:singlepricing} with the following four baseline algorithms:
\begin{description} 
  \setlength{\parskip}{0pt}
  \setlength{\itemsep}{0pt}
\item[Selling all items.] Assign $X = V$ and price $p(v) = f(V) - f(V \setminus v)$ for each $v \in V$. This algorithm gives a stable assignment.
\item[Random pricing.] Price $p(v) \in [0, f(v)]$ uniformly at random for each $v \in V$ and find an assignment $X$ by the greedy algorithm.
\item[Scaled pricing.] Price $p(v) = \alpha f(v)$ for each $v$ and find an assignment $X$ by the greedy algorithm. $\alpha$ is chosen optimally from $\{0.1, \ldots, 1.0\}$.
\item[Ascending pricing.] Start from $X = V$ and $p(v) = 0 \ (v \in V)$, repeat the following process: Price $p(v) = \min_{X: x \in X}(f(X) - f(X \setminus x))$ for each $v \in X$, remove $\tilde x$ that attains the minimum from $X$, and then price $p(\tilde x) = +\infty$. This algorithm is motivated by the ascending auction~\cite{krishna2009auction}.
\end{description}
We remark that there are no existing algorithms that are directly applicable to the optimal pricing problem with submodular valuations (see also Section \ref{sec:relatedwork}). 

We used all the networks described above; we set $|V| = 100$, $|W| = 10000$, $d = 10$ and $q_{\text{max}} = 0.3$ for \texttt{Uniform} and \texttt{PowerLaw}.
The result is summarized in Table~\ref{tbl:comparison}.
The proposed algorithm outperforms all compared algorithms for all datasets

\vspace{-1em}
\paragraph{(3) Scalability}

We evaluated the scalability of the proposed algorithm.
We used \texttt{Uniform} with $|V| \in \{16, 32, \ldots, 1024\}$, $|W| \in \{100, 1000, 10000, 100000\}$, $d = 10$, and $q_{\text{max}} = 0.3$.
We also conducted the same experiment on \texttt{PowerLaw} but we omit it since it yields very similar results.

The result is shown in Figure~\ref{fig:scalability}.
The elapsed times were (roughly) proportional to both $|V|$ and $|W|$.
This is consistent with our analysis that the proposed algorithm runs in $O(|V| |E|)$ time, and the number of edges is proportional to $|W|$ for these networks.
Therefore, the proposed algorithm scales to moderately large networks.

\vspace{-1em}
\paragraph{(4) Number of allocated channels and activation probabilities}

We observe the relationship between activation probability and the obtained allocation.
We used \texttt{Uniform} and \texttt{PowerLaw} with the parameters $|V| = 100$,  $|W| = 10000$, and $d = 10$.
We controlled the maximum activation probability $q_{\text{max}} \in \{0.05, 0.10, \ldots, 0.95\}$ and observe the number of assigned marketing channels.

The result is shown in Figure~\ref{fig:probability}.
For both networks, when $q_{\text{max}}$ was small the proposed algorithm assigned all channels, and when $q_{\text{max}}$ was large it assigned a few channels.
The number of assigned channels decreased much faster in \texttt{PowerLaw} than in \texttt{Uniform}, since there were highly correlated marketing channels in \texttt{PowerLaw}.

\vspace{-1em}
\paragraph{(5) Profit and the number of advertisers}

Next, we conducted experiments on the multiple advertisers case and the multiple collaborating advertisers case.
Here we observe the relationship between the profit and the number of advertisers in these settings.
For these experiments, we modified \texttt{Uniform} and \texttt{PowerLaw}
to assign multiple probabilities $q_1(e), \ldots, q_n(e)$ for each edge, each of which follows the uniform distribution on $[0, q_{\text{max}}]$.

The result is shown in Figure~\ref{fig:advertisers}.
By comparing two results obtained by the multiple (non-collaborating) advertisers case, the number of advertisers had little influence on the profit.
On the other hand, by comparing the results obtained by collaborating advertisers, the profit increased when the number of advertisers increased.
Moreover, the profits obtained from the collaborating advertisers consistently outperformed those obtained from non-collaborating advertisers.
This means that collaboration of advertisers yields a better profit to the publisher.
Note that we could not observe the difference between \texttt{Uniform} and \texttt{PowerLaw}.

\begin{figure}[t]
\hspace{-1em}
\begin{tabular}{cc}
\begin{minipage}{0.495\hsize}
\centering
\resizebox{\hsize}{!}{\includegraphics{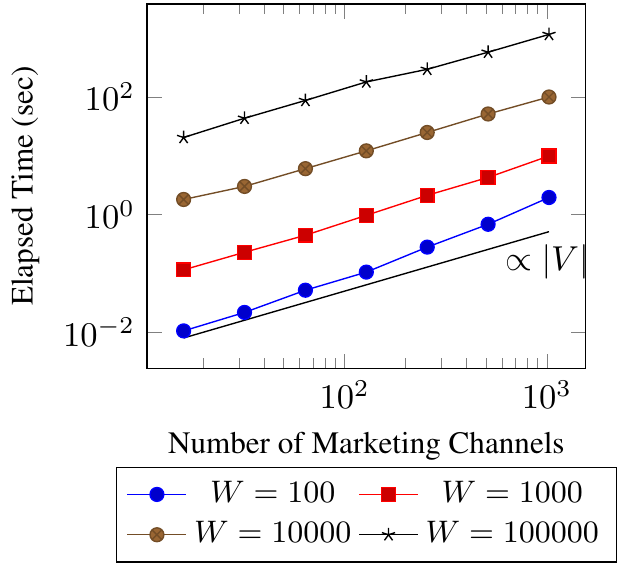}}
\caption{Scalability of the proposed algorithm.}
\label{fig:scalability}
\end{minipage}\quad
\begin{minipage}{0.495\hsize}
\centering
\resizebox{\hsize}{!}{\includegraphics{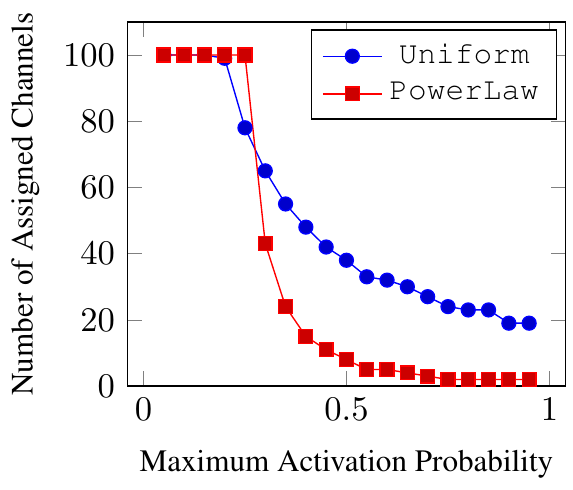}}
\caption{The number of assigned channels versus edge probability.}
\label{fig:probability}
\end{minipage}\quad
\end{tabular}
\end{figure}
\begin{figure}[t]
\centering
\resizebox{0.495\hsize}{!}{\includegraphics{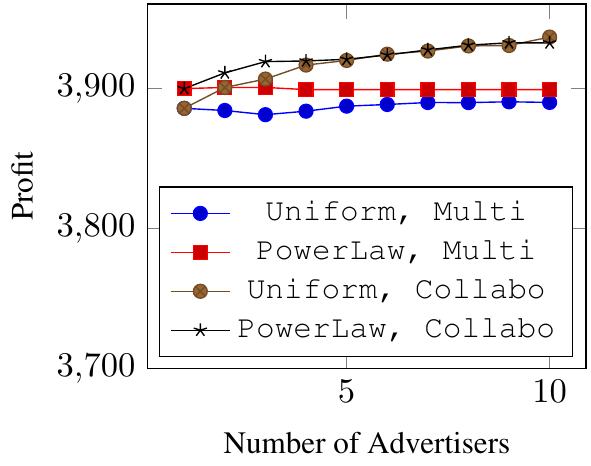}}
\caption{Profit versus the number of non-collaborating or collaborating advertisers.}
\label{fig:advertisers}
\vspace{-1em}
\end{figure}

\section{Conclusion}
\label{sec:conclusion}

We propose some future works.
One is to develop an approximate pricing algorithm for the case that multiple buyers have limited budgets.
Another one is to analyze the optimal pricing problem with multiple sellers.
In this study, we assumed there is one seller; who can be regarded as a monopolist.
The seller can select both the assignment and the price as long as they satisfy the stability condition.
This situation is highly advantageous for the seller.
Finally, in this study, we do not consider nonlinear or non-anonymous pricing.
It would be interesting in future work to analyze the effect of such generalizations of pricing.

\medskip

\noindent \textbf{Acknowledgments} This work was supported by JSPS KAKENHI Grant Number 16K16011, and JST, ERATO, Kawarabayashi Large Graph Project.

\bibliographystyle{aaai}
\bibliography{main}

\clearpage

\appendix

\section{Extension to limited budget cases}

In this section, we extend our algorithms for the optimal pricing problem with unlimited budget cases to ones for the limited budget cases.

First, suppose that there exists only one buyer who has a limited budget $B$. 
Since his payment is limited by this budget, the output $(p, X)$ of Algorithm \ref{alg:singlepricing} may fail of the budget constraint. 
Thus, we arbitrarily discounts the prices $p(x)$ for $x \in X$ so that $p(X) = B$ holds, and then returns $(p, X)$. 

\OMIT{
\begin{algorithm}[tb]
\caption{Pricing algorithm for a single buyer with a budget constraint.}
\begin{algorithmic}
  \FOR{$s = 1, 2, \ldots, |V|$}
  \STATE{Let $X^s$ be the largest $s$ elements of $f(x)$}
  \STATE{Price $p^s(x) = f(X^s) - f(X^s \setminus x)$ for all $x \in X^s$ and $p^s(y) = +\infty$ for all $y \in V \setminus X^s$}
  \ENDFOR
  \STATE{Select $p$ and $X$ that attains maximum of $p^s(X^s)$}
  \IF{$p(X) > B$}
  \STATE{Discount $p(x)$ \ ($x \in X$) arbitrary to $p(X) = B$}
  \ENDIF
  \STATE{\textbf{Return $p$ and $X$}}
\end{algorithmic}
\end{algorithm}
}

\begin{algorithm}[tb]
\caption{Pricing algorithm for a single buyer with a budget constraint.}
\label{alg:singlebudget}
\begin{algorithmic}
  \STATE{Let $p$ and $X$ be the outputs of Algorithm \ref{alg:singlepricing}}
  \STATE{\textbf{if} $p(X) > B$ \textbf{then} discount $p(x)$ \ ($x \in X$) arbitrary so that $p(X) = B$}
  \STATE{\textbf{Return $p$ and $X$}}
\end{algorithmic}
\end{algorithm}

\begin{theorem}
\label{thm:singlebudget}
Let $(p^*, X^*)$ be an optimal solution to \eqref{eq:problemmultiple},
and let $(p, X)$ be the solution obtained by Algorithm~\ref{alg:singlebudget}.
It holds that $X\in D(p)$ and \((1-\kappa(|X^*|))p^*(X^*)\le p(X)\).
\end{theorem}
\begin{proof}
If no discount has been performed, then we obtain the stability of $(p, X)$ and the claimed theoretical guarantee by the same analysis as Theorem~\ref{thm:singlepricing}. 

We assume that a discount has been performed. 
Let $(p, X)$ be the output of Algorithm \ref{alg:singlepricing}, and $p'$ be the price vector obtained by Algorithm \ref{alg:singlebudget}. 
The profit is $B$, and this is the optimal value. 
It remains to show that $X \in D(p')$. 
Let $Y$ be an arbitrary subset of $V$. 
If $Y \not\subseteq X$, then it holds that $f(X) - p'(X) > f(Y)-p'(Y) = - \infty$. 
Thus, we may assume that $Y \subseteq X$. 
It follows that
\begin{align*}
&(f(X) - p'(X)) - (f(Y) - p'(Y)) \\
&= f(X)-f(Y)-p'(X \setminus Y) \\
&\geq f(X)-f(Y)-p'(X \setminus Y)\\
&= (f(X) - p'(X)) - (f(Y) - p'(Y))\\
&\geq 0.
\end{align*}
Therefore, $X \in D(p')$ holds. 
This completes the proof.
\end{proof}

Suppose that there exist multiple collaborating buyers, who share a budget $B$. 
The total payment of buyers must not exceed $B$. 
The idea of Algorithm \ref{alg:singlebudget} works for this case, by executing Algorithm \ref{alg:collaboratepricing} instead of Algorithm \ref{alg:singlepricing}. 
Because the proof of Theorem \ref{thm:singlebudget} does not use the submodularity of $f$, we can derive the same result as Theorem \ref{thm:singlebudget} by a similar result using Theorem \ref{thm:collaborating} instead of Theorem \ref{thm:singlepricing}. 
\begin{corollary}
Let $(p^*, X^*)$ be an optimal solution to \eqref{eq:problemcollaborating},
and let $(p, X)$ be the solution obtained by Algorithm~\ref{alg:singlebudget}.
It holds that $X\in D(p)$ and \((1-\kappa(|X^*|))p^*(X^*)\le p(X)\).
\end{corollary}

\section{Omitted results and proofs}

In this section, we provide omitted theorems and prove all results. 

\begin{proof}[Proof of Lemma~\ref{lem:curvaturebound}]
By a direct calculation, we have
\begin{align}
  \frac{f_i(X) - f_i(X \setminus x)}{f_i(x)}
  &= \frac{\displaystyle \sum_{w \in N(x)} q_i(x, w) (1 - Q_i(X \setminus x, w))}{\displaystyle \sum_{w \in N(x)} q_i(x, w)} \nonumber \\
  &\ge 1 - \max_{w \in N(x)} Q_i(X \setminus x, w),
\end{align}
where $N(x) = \{ w \mid (x, w) \in E \}$. 
By taking the minimum over $x \in X$ and $|X| = s$, we obtain this result.
\end{proof}

\begin{proof}[Proof of Theorem~\ref{thm:curvature}]
Since $q_i(e) \le q$, we have
\begin{align}
  Q_i(X \setminus x, w) \le 1 - (1 - q)^{\min\{|X|-1, \mathrm{deg}(w)-1\}}
\end{align}
for any $X$, $x\in X$, and $w$.
Thus we obtain the result.
\end{proof}

\begin{proof}[Proof of Lemma~\ref{lem:single_obs}]
Let $p$ be the price vector defined as in the statement. 
We first prove $X \in D(p)$.
Note that for any $Y \not\subseteq X$, it holds that $f(X) - p(X) > f(Y) - p(Y) = -\infty$ by definition of $p$. 
For any \(Y\subseteq X\) and \(x \in Y\), we have
\begin{align}
  f(Y) - p(Y)
  &= (f(Y) - f(Y\setminus x)) + f(Y\setminus x) - p(Y)\notag\\
  &\ge (f(X) - f(X\setminus x)) + f(Y\setminus x) - p(Y)\notag\\
  &= f(Y\setminus x) - p(Y\setminus x),
\end{align}
because $f$ is a submodular function and $p(x)=f(X) - f(X\setminus x)$.
Thus, it holds that \(f(X)-p(X)\ge f(Y)-p(Y)\) for all \(Y \subseteq X\), which means $X \in D(p)$.

Moreover, for any price vector $p'$ with $X \in D(p')$, we have 
\begin{align}
  f(X) - p'(X) \ge f(X \setminus x) - p'(X \setminus x)
\end{align}
for all $x \in X$. 
Thus, it holds that
\begin{align}
  p'(x) \le f(X) - f(X \setminus x) \quad (x \in X), 
\end{align}
and we have
\begin{align}
  p'(X)
  &\le \sum_{x \in X} (f(X) - f(X \setminus x)) = p(X)\notag.
\end{align}
Therefore, the lemma holds.
\end{proof}

\begin{proof}[Proof of Lemma~\ref{lem:Dtoh}]
Lemma \ref{lem:single_obs} implies that the maximum objective value of \eqref{eq:problemmultiple} for a fixed assignment $X$ is $\sum_{x \in X} (f(X) - f(X \setminus x))$. 
Thus for any optimal solution $(p^*, X^*)$ to \eqref{eq:problemmultiple}, $X^*$ attains \eqref{eq:problemsingle_eq}. 
Conversely, any assignment achieving \eqref{eq:problemsingle_eq} together with the price vector defined in Lemma \ref{lem:single_obs} is an optimal solution to \eqref{eq:problemmultiple}. 
This proves the lemma.
\end{proof}

\begin{proof}[Proof of Theorem~\ref{thm:NPhard}]
Let $\phi$ be an instance of the one-in-three positive 3-SAT problem with the set $V$ of variables and the set $W$ of clauses. 
We construct a bipartite graph $G=(V \cup W, E)$, where $G$ has an edge $(v, w) \in E$ if and only if variable $v$ is contained in clause $w$. 
We define $q(e)=1$ for all $e \in E$. 
Note that we have $\mathrm{deg}(w) = 3$ for all $w \in W$. 
Let $f$ be a submodular function defined by \eqref{eq:fQ} with $\gamma = 1$.

We observe that $f(X)$ is the number of right vertices covered by $X$, because we have $Q(X, w) = 1$ if $(x, w) \in E$ for some $x \in X$, and $Q(X, w) = 0$ otherwise. 
Thus, $h(X) = \sum_{x \in X} (f(X) - f(X \setminus x))$ is the number of right vertices that is covered by $X$ \emph{exactly once}. 
Therefore, it holds that $\phi$ is satisfiable if and only if there exists $X$ such that $h(X) = |W|$.
\end{proof}

\begin{proof}[Proof of Theorem~\ref{thm:exporacle}]
Let $X^* \subseteq V$ be a fixed assignment. We define $f$ by 
\begin{align}
  f(X) = \begin{cases}
    2|X|, & |X| < |X^*|, \\
    2|X^*| - 1, & |X| = |X^*| \text{ and } X \neq X^*, \\
    2|X^*|, & \text{otherwise}.
  \end{cases}
\end{align}
This function is a monotone nondecreasing submodular function.
For this function, the optimal solution to \eqref{eq:problemmultiple} is given by a price vector $p^*$ defined as
\begin{align}
  p^*(x) = \begin{cases}
    2, & x \in X^*, \\
    +\infty, & \text{otherwise,}
  \end{cases}
\end{align}
and $X^*$.
On the other hand, we need 
${|V| \choose |X^*|}$ 
oracle evaluations to distinguish $X^*$ and other $X$ with $|X| = |X^*|$.
\end{proof}

\begin{proof}[Proof of Theorem~\ref{thm:singlepricing}]
Since we have $X \in D(p)$ by Lemma \ref{lem:single_obs},
it suffices to prove that  \((1 - \kappa(|X^*|)) p^*(X^*) \le p(X)\).
Since $X^* \in D(p^*)$, we have
\begin{align}
  f(X^*) - p^*(X^*) \ge f(X^* \setminus x) - p^*(X^* \setminus x),
\end{align}
for all $x \in X^*$, which implies that 
\begin{align}
  p^*(x) \le f(X^*) - f(X^* \setminus x) \quad (x \in X^*).
\end{align}
Let $(p^s, X^s)$ be the solution produced by the algorithm for $s = |X^*|$.
Then we have
\begin{align}
  p^*(X^*) 
  &\le \sum_{x \in X^*} \left( f(X^*) - f(X^* \setminus x) \right) \nonumber \\
  &\le \sum_{x \in X^*} f(x) \le \sum_{x \in X^s} f(x) \nonumber \\
  &\le \frac{1}{1 - \kappa(s)} \sum_{x \in X^s} \left( f(X^s) - f(X^s \setminus x) \right) \nonumber \\
  &= \frac{1}{1 - \kappa(s)} p^s(X^s) \le \frac{1}{1 - \kappa(s)} p(X),
  \end{align}
  which proves the theorem.
\end{proof}

\begin{remark}
Since many greedy algorithms usually chooses an item maximizing the marginal gain at each step, one may think of such a greedy algorithm for the optimal pricing problem (with a single buyer): starting with $X=\emptyset$,
repeatedly add the item $x\in V\setminus X$ maximizing the marginal gain to $X$ (i.e., $f(X\cup\{x\})-f(X)$) while the gain is positive,
and then outputs $(p, X)$ where $p$ is defined by using Lemma \ref{lem:single_obs}. 
However, we only obtained an approximation guarantee of $(1-\kappa(|X|))^2$, which is worse than the guarantee of Algorithm~\ref{alg:singlepricing}. 
\end{remark}

\begin{theorem}
\label{thm:muchworse}
There exists a monotone nondecreasing submodular function $f:2^V\to\mathbb{R}$ such that 
\begin{align}
\frac{\max_{X} f(X)}{\max_{X} h(X)}\ge H(|V|),
\end{align}
where \(H(k)\) is the $k$th harmonic number (i.e., \(H(k)=\sum_{i=1}^k 1/i\)).
\end{theorem}
\begin{proof}
  Let \(f(X)=H(|X|)\).
  Then, \(\max_{X}f(X)=f(|V|)=H(|V|)\) and
  \(h(X)=\sum_{x\in X}(f(X)-f(X\setminus x))=1\) for all \(X\).
 Thus, we have obtained our proposition.
\end{proof}

\begin{theorem}
\label{thm:worsebutcurvature}
For any monotone nondecreasing submodular function $f:2^V\to\mathbb{R}$ with curvature $\kappa$,
we have 
\begin{align}
\frac{\max_{X} f(X)}{\max_{X} h(X)}\le \frac{1}{1-\kappa(|V|)}.
\end{align}
\end{theorem}
\begin{proof}
This proof can be done in a straightforward manner
\begin{align}
  \frac{\max_{X} f(X)}{\max_{X} h(X)}
  &\le \frac{f(V)}{h(V)}\notag\\
  &\le \frac{\sum_{x\in V}f(x)}{(1-\kappa(|V|))\sum_{x\in V}f(x)}\notag\\
  &= \frac{1}{1-\kappa(|V|)}.
\end{align}
\end{proof}


\begin{proof}[Proof of Theorem~\ref{thm:stablepricing}]
  We give a reduction from the \emph{exact cover by 3-sets problem (X3C)}, which is an NP-complete problem~\cite{garey1979computers}.
  In this problem, we are given a set $E$ with $|E|=3l$ and a collection $\mathcal{C}=\{C_1,\dots,C_m\}$ of 3-element subsets of $E$,
  and ask whether or not $\mathcal{C}$ contains an exact cover for $E$, i.e., an exact cover is 
  a subcollection $\mathcal{C}'\subseteq \mathcal{C}$ such that every element for $E$ occur in exactly one member of $\mathcal{C}'$.

  Let $n=2$, $V=\{0,1,\dots,m\}$, $C_0=E$,
  $X_1=\{0\}$, $X_2=\{1,\dots,m\}$, and $f_1(X)=|\bigcup_{i\in X}C_i|$.
  Let $f_2(X)=|X|$ if $0\not\in X$ and $f_2(X)=|X|+l$ otherwise.
  Note that $f_1$ is monotone submodular and $f_2$ is monotone linear.

  We first claim that if there exists a stable price vector,
  the following price vector $p^*$ is also stable; $p^*(0)=l+1$ and $p^*(i)=1$ for $i\in X_2$.
  Since $f_2$ is linear, the condition \(X_2\in D_2(p)\) holds if and only if $p(0)\le l+1$ and $p(i)\ge 1$ for $i\in X_2$.
  In addition, if \(X_1\in D_1(p)\) holds for some price vector $p$,
  \(X_1\in D_1(p')\) also holds for any price vector \(p'\) such that $p'(0)\ge p(0)$ and $p'(i)\le p(i)$ for $i\in X_2$.
  Thus, the assignment $X_1,X_2$ has a stable price vector if and only if \(X_1\in D_1(p^*)\).
  Note that \(f_1(X_1)-p^*(X_1)=3l-(l+1)=2l-1\).

  Next, we show that \(X_1\not\in D_1(p^*)\) holds if and only if 
  there exists a solution for the X3C problem.
  Assume that there exists an exact cover $\mathcal{C}'\subseteq \mathcal{C}$ for $E$.
  Let \(X=\{i\mid C_i\in\mathcal{C}'\}\).
  Then, \(|X|=l\) by the definition of \(\mathcal{C}'\) and \(f_1(X)-p^*(X)=3l-|X|=2l\).
  Thus, we have \(X_1\not\in D_1(p^*)\).
  Conversely, assume that \(X_1\not\in D_1(p^*)\) and \(X_1'\in D_1(p^*)\).
  Then, \(f_1(X_1')-p^*(X_1')>2l-1\).
  Here, it is not difficult to see that \(f_1(X)-p^*(X)>2l-1\) holds only when \(0\not\in X\), \(|\sum_{i\in X}C_i|=3l\), and \(|X|=l\).
  Hence, \(\{C_i\mid i\in X\}\) is a solution for the X3C problem.

  Therefore, the assignment $X_1,X_2$ has a stable price vector
  if and only if there exists no solution for
  the X3C problem.
  This implies the coNP-hardness of the problem.  
\end{proof}

\begin{theorem}
\label{thm:NPhardmultiple}
  It is NP-hard to determine,
  for a given price vector $p$,
  the existence of an assignment \(X_1,\dots,X_n\) such that
  $X_i\in D_i(p)$ for all $i\in N$.
\end{theorem}
\begin{proof}[Proof of Theorem~\ref{thm:NPhardmultiple}]
  We show that the \emph{partition} problem can be reduced to the problem.
  Here the partition problem is the following NP-complete problem~\cite{garey1979computers}:
  given $m$ positive integers $a_1,\dots,a_m$ with \(\sum_{i=1}^n a_i=2T\),
  determine whether there exists a subset of $I\subseteq \{1,\dots,n\}$
  such that \(\sum_{i\in I}a_i=T\).
  Let $n=2$, $V=\{1,\dots,m\}$, $p(x)=0~(\forall x\in V)$, and $f_1(X)=f_2(X)=\min\{T,\sum_{i\in X}a_i\}$.
  Note that $f_1$ and $f_2$ are monotone submodular functions.
  Then, \(D_1(p)=D_2(p)=\{X\subseteq V\mid \sum_{x\in X}a_x\ge T\}\).
  Thus, there exists an assignment \(X_1,X_2\) such that \(X_1\in D_1(p)\) and \(X_2\in D_2(p)\)
  if and only if there exists a desired partition \(I\).
  Therefore, we obtain the theorem.
\end{proof}

\begin{proof}[Proof of Theorem~\ref{thm:multiplepricing}]
  We first prove the $(1-\kappa(s))$-stable condition
  by showing that \(f_i(X_i)-p(X_i)\ge 0\ge (1-\kappa_i(s))f_i(X)-p(X)\) holds for any \(X\subseteq V\) and \(i\in N\).
  The left inequality holds by
\begin{align}
  f_i(X_i) - p(X_i) 
  &= f_i(X_i) - \sum_{x \in X_i} (f_i(X^s) - f_i(X^s \setminus x)) \nonumber \\
  &\ge f_i(X_i) - \sum_{x \in X_i} (f_i(X_i) - f_i(X_i \setminus x)) \nonumber \\
  &\ge 0,
\end{align}
and we show the right inequality.
If $X \not \subseteq X^s = X_1 \cup \cdots \cup X_n$ then $(1-\kappa_i(s))f_i(X) - p(X) = -\infty$.
Thus, we can assume that \(X\subseteq X^s\) and we have
\begin{align}
  f_i(X) - p(X) 
  &= f_i(X) - \sum_{x \in X} \max_j (f_j(X^s) - f_j(X^s \setminus x)) \nonumber \\
  &\le f_i(X) - \sum_{x \in X} (f_i(X^s) - f_i(X^s \setminus x)) \nonumber \\
  &\le f_i(X) - (1 - \kappa_i(s)) \sum_{x \in X} f_i(x) \nonumber \\
  &\le f_i(X) - (1 - \kappa_i(s)) f_i(X) \nonumber \\
  &= \kappa_i(s) f_i(X).
\end{align}
Therefore, we have
\begin{align}
  \left(1 - \kappa_i(s)\right) f_i(X) - p(X) \le 0,
\end{align}
and \((p,\mathbf{X})\) is $(1-\kappa(s))$-stable.

Next, we show the approximate profit guarantee.
Since \((p^*,\mathbf{X}^*)\) is stable,
we have
\begin{align}
  p^*(X_i^*) 
  &\le 
  \sum_{x \in X_i^*} (f_i(X_i^*) - f_i(X_i^* \setminus x))
  \le \sum_{x \in X_i^*} f_i(x).
\end{align}
Thus, we obtain
\begin{align}
  \sum_{i\in N} p^*(X_i^*)
  &\le \sum_{i\in N} \sum_{x \in X_i^*} f_i(x) \nonumber \\
  &\le \sum_{i\in N} \sum_{x \in X_i^*} \max_{j\in N} f_j(x) \nonumber \\
  &\le \sum_{x \in X^{s^*}} \max_{j\in N} f_j(x) \nonumber \\
  &\le \sum_{x \in X^{s^*}} \max_{j\in N} \frac{f_j(X^{s^*})-f_j(X^{s^*}\setminus x)}{1-\kappa(s^*)} \nonumber \\
  &\le \sum_{x \in X^{s^*}} \frac{f_i(X^{s^*})-f_i(X^{s^*}\setminus x)}{1-\kappa(s^*)} \nonumber \\
  &= \frac{p(X^{s^*})}{1-\kappa(s^*)} \le \frac{p(X)}{1-\kappa(s^*)}.
\end{align}
\end{proof}

\begin{example}
\label{ex:nonsubmodular}
\begin{table}[ht]
\centering
\caption{An example that the aggregate valuation function is not submodular.}
\label{tab:nonsubmodular}
{\tabcolsep = 1.5mm
\begin{tabular}{c|ccccccc}
\hline
      & $\{a\}$ & $\{b\}$ & $\{c\}$ & $\{a,b\}$ & $\{b,c\}$ & $\{a,c\}$ & $\{a,b,c\}$\\\hline
$f_1$ & $1$     & $1$     & $1$     & $2$       & $1$       & $2$       & $2$  \\\hline
$f_2$ & $2$     & $2$     & $2$     & $3$       & $4$       & $3$       & $4$  \\\hline
$f$   & $2$     & $2$     & $2$     & $3$       & $4$       & $3$       & $5$  \\\hline
\end{tabular}
}
\end{table}

Let $f_1,f_2$ be valuation functions
and $f$ be the aggregated valuation function
defined as in Table \ref{tab:nonsubmodular},
Then, $f_1$ and $f_2$ are submodular but $f(\{a,b\}) + f(\{c,a\}) \ge f(\{a,b,c\}) + f(\{a\})$ does not hold.
\end{example}

\begin{proof}[Proof of Lemma~\ref{lem:collaborating_obs}]
Let $p$ be the price vector defined as the statement. 
We first prove $X \in D(p)$.
Note that for any $Y \not\subseteq X$, it holds that $f(X) - p(X) > f(Y) - p(Y) = -\infty$ by definition of $p$. 
For any \(Y\subseteq X\) and \(x\in Y\), we have
\begin{align}
  &(f(Y)-p(Y))-(f(Y\setminus x)-p(Y\setminus x))\notag\\
  &= (f(Y) - f(Y\setminus x)) - p(x)\ge 0,
\end{align}
by the definition of \(p(x)\).
Thus, it holds that \(f(X)-p(X)\ge f(Y)-p(Y)\) for all \(Y \subseteq X\), which means $X \in D(p)$.

Moreover, for any price vector $p'$ with $X \in D(p')$, we have 
\begin{align}
  f(Y) - p'(Y) \ge f(Y \setminus x) - p'(Y \setminus x)
\end{align}
for all $Y\subseteq X$ and $x \in Y$. 
Thus, it holds that
\begin{align}
  p'(x) \le f(Y) - f(Y \setminus x),
\end{align}
and we have
\begin{align}
  p'(X)
  &\le \sum_{x \in X} \min_{Y\subseteq X: x\in Y}(f(Y) - f(Y \setminus x)) = p(X).
\end{align}
Therefore, the lemma holds.
\end{proof}

\begin{proof}[Proof of Lemma~\ref{lem:collaborating 1}]
Let $(X^*_1, \ldots, X^*_n)$ (resp., $(X'_1, \ldots, X'_n)$) be the sets which form a partition of $X$ (resp., $X \setminus x$) attaining the maximum in the aggregated valuation function. 
Suppose that $x \in X_j^*$.
By replacing $X_i'$ with $X^*_i \setminus x$ for $i\in N$, we have
\begin{align}
  f(X) - f(X \setminus x) 
  & = \sum_{i\in N} f_i(X^*_i) - \sum_{i\in N} f_i(X'_i) \nonumber\\
  &\le f_j(X^*_j) - f_j(X^*_j \setminus x) \nonumber \\
                          &\le f_j(x)  \le \max_j f_j(x)=f(x).
\end{align}
This proves the lemma.
\end{proof}

\begin{proof}[Proof of Lemma~\ref{lem:collaborating 2}]
Let $X'_1, \ldots, X'_n$ be the sets which form a partition of $X \setminus x$ 
satisfying $f(X \setminus x) = \sum_{i\in N} f_i(X'_i)$.
Suppose that \(f_j(x)= f(x)\).
By setting $X_j = X_j' \cup \{x\}$ and $X_k = X_k'$ for $k \neq j$, we have
\begin{align}
  f(X) - f(X \setminus x) 
  & \ge \sum_{i\in N} f_i(X_i) - \sum_{i\in N} f_i(X'_i) \nonumber\\ 
  &\ge f_j(X_j) - f_j(X_j\setminus x) \nonumber \\
  &= \frac{f_j(X_j) - f_j(X_j\setminus x)}{f_j(x)}f(x) \nonumber \\
  &\ge \min_{i\in N}\frac{f_i(X) - f_i(X\setminus x)}{f_i(x)}f(x) \nonumber \\
  &\ge (1-\max_{i\in N}\kappa_i(|X|))f(x) \nonumber \\
  &= (1-\kappa(|X|))f(x),
\end{align}
which proves the lemma.
\end{proof}

\begin{proof}[Proof of Theorem \ref{thm:collaborating}]
  First, we show the stable condition $X\in D(p)$, i.e.,
  \(f(X)-p(X)\ge f(Y)-p(Y)\) for any $Y\subseteq V$. 
  By definition of $p$, we may only consider $Y \subseteq X$.
  For any $x \in Y$, we have
\begin{align}
  &(f(Y) - p(Y)) - (f(Y \setminus x) - p(Y \setminus x)) \nonumber \\
  &=(f(Y) - f(Y \setminus x)) - p(x) \nonumber \\
  &\ge\min_{i\in N}\frac{f_i(Y)-f_i(Y\setminus x)}{f_i(x)}f(x) - p(x) \nonumber \\
  &\ge\min_{i\in N}\frac{f_i(X)-f_i(X\setminus x)}{f_i(x)}f(x) - p(x) = 0,
\end{align}
where the first inequality holds by Lemma \ref{lem:collaborating 2}.
Thus, we have \(X\in D(p)\).

Next, we show the approximation guarantee. 
Since $(p^*, X^*)$ satisfies $X^* \in D(p^*)$, we have
\begin{align}
  f(X^*) - p^*(X^*) \ge f(X^* \setminus x) - p^*(X \setminus x) \quad (x \in X^*).
\end{align}
Therefore, it holds that
\begin{align}
  p^*(x) \le f(X^*) - f(X^* \setminus x).
\end{align}
for all $x \in X^*$.
We denote $s = |X^*|$, and let $(p^s, X^s)$ be the solution produced by the algorithm for $s$.
We then have
\begin{align}
  p^*(X^*) 
  &\le \sum_{x \in X^*} \left( f(X^*) - f(X^* \setminus x) \right) \nonumber \\
  &\le \sum_{x \in X^*} f(x) \le \sum_{x \in X^s} f(x) \nonumber \\
  &= \frac{p^s(X^s)}{\min_{i\in N}\frac{f_i(X^s)-f_i(X^s\setminus x)}{f_i(x)}} \nonumber\\
  &\le \frac{p^s(X^s)}{1 - \kappa(s)}  \le \frac{p(X)}{1 - \kappa(s)}, 
\end{align}
where the second and the forth inequalities follow from Lemma \ref{lem:collaborating 1} and Lemma \ref{lem:collaborating 2}, respectively. 
This shows the theorem.
\end{proof}

\section{Description of datasets used in the paper}

The following is a description of datasets.

\begin{description}[font=\tt,leftmargin=10pt]
  \setlength{\parskip}{0pt}
  \setlength{\itemsep}{0pt}
  \item[Uniform] is a random bipartite network with $|V|$ left vertices and $|W|$ right vertices, where the degree of right vertices are constant ($= d$) and the degree distribution of the left vertices follows the uniform distribution. We set the activation probability $q(e) \in [0, q_{\text{max}}]$. We control these parameters to observe the performance of the algorithm.

  \item[PowerLaw] is similar to \texttt{Uniform}. The only difference is that the degree distribution of the left vertices follows a power-law distribution with exponent $2$.
\end{description}
\begin{description}[font=\tt,leftmargin=10pt]
  \item[Last.fm] is constructed from Last.fm 1K Dataset\footnote{\url{http://www.dtic.upf.edu/~ocelma/MusicRecommendationDataset/lastfm-1K.html}}, which contains 992 users' listening log on songs.
    We select top 1,000 frequently played songs and construct a bipartite graph, which has 1,000 left vertices (songs) and 986 right vertices (users) with 1,832,088 (multiple) edges.
    We set activation probability as $q(e) = 0.01$.

  \item[MovieLens] is constructed from MovieLens 10M Dataset\footnote{\url{http://files.grouplens.org/datasets/movielens/ml-10m-README.html}}, which contains 100,00,054 ratings ($1$ to $5$) to 10,681 movies by 71,567 users.
    We select top 1,000 frequently rated movies and construct a bipartite graph, which has 1,000 left vertices (movies) and 10,585 right vertices (users) with 1,357,805 edges.
    We set activation probability as $q(e) = 0.02 \times \text{rating}$.

  \item[BookCrossing] is constructed from Book-Crossing dataset\footnote{\url{http://www2.informatik.uni-freiburg.de/~cziegler/BX/}}, which contains 
    1149780 ratings (0 to 10) for 271,379 books given by 278,858 users.
    We select top 1,000 frequently rated books and constructed the network, which has 1,000 left vertices and 35,634 right vertices with 162,767 edges.
    We set activation probability as $q(e) = 0.01 \times \text{rating} + 0.01$.
\end{description}

\end{document}